\documentclass[11pt]{article}
\usepackage[a4paper]{geometry}
\usepackage{amsfonts, amsmath, amssymb, amsthm, graphicx, caption, authblk, multirow, makecell, framed, float, xcolor, enumitem, tikz, hyperref}
\usepackage{algorithm,algorithmicx,algpseudocode,mathrsfs}
\usetikzlibrary{fit,shapes.arrows}
\theoremstyle{definition}
\newtheorem{theorem}{Theorem}

\theoremstyle{remark}

\newcommand{\sh}[1]{\textsc{#1}}
\newcommand{\id}{\text{id}}

\newcommand*{\mybox}[1]{%
  \framebox{\raisebox{0cm}[0.5\baselineskip][0.05\baselineskip]{%
    \hbox to 0.10cm {\hss#1\hss}}}\hspace{0.05cm}}

\begin{document}
\title{A Complexity Hierarchy of Shuffles in Card-Based Protocols}
\author[1]{Tomoki Ono\thanks{\texttt{onotom@uec.ac.jp}}}
\author[2]{Suthee Ruangwises\thanks{\texttt{suthee@cp.eng.chula.ac.th}}}
\affil[1]{The University of Electro-Communications, Tokyo, Japan}
\affil[2]{Department of Computer Engineering, Faculty of Engineering, Chulalongkorn University, Bangkok, Thailand}
\date{}
\maketitle

\begin{abstract}
Card-based cryptography uses physical playing cards to construct protocols for secure multi-party computation. Existing card-based protocols employ various types of shuffles, some of which are easy to implement in practice while others are considerably more complex. In this paper, we classify shuffle operations into several levels according to their implementation complexity. We motivate this hierarchy from both practical and theoretical perspectives, and prove separation results between several levels by showing that certain shuffles cannot be realized using only operations from lower levels. Finally, we propose a new complexity measure for evaluating card-based protocols based on this hierarchy.

\textbf{Keywords:} card-based cryptography, secure multi-party computation, shuffles, complexity, permutation sets
\end{abstract}

\section{Introduction}
Card-based cryptography is a branch of cryptography that uses physical playing cards to construct cryptographic protocols. Instead of relying on computers or electronic devices, the computation is carried out through simple physical operations on cards, allowing participants to perform cryptographic tasks while keeping their private information hidden. Applications of card-based cryptography include secure multi-party computation, zero-knowledge proofs, and other cryptographic primitives. The study of card-based cryptography dates back to 1990, when den Boer~\cite{5card} introduced the ``five-card trick'', a protocol that allows two parties to compute the logical AND of their input bits using five cards.

A key component of card-based protocols is the use of \emph{shuffles}. Shuffles are randomized operations applied to a deck of cards in order to hide information about the positions of the cards. They play a crucial role in ensuring the security of card-based protocols, as they prevent players from learning sensitive information during the execution of the protocol. Over the years, a variety of shuffles have been proposed and used in the literature.

However, not all shuffle operations are equally practical. Some shuffles are easy to perform in practice and can be implemented naturally by human players, while others are more difficult to realize physically. As a result, most existing card-based protocols rely on only a small number of fundamental shuffle primitives that are both practically implementable and widely used in protocol design.

In particular, the vast majority of known protocols use shuffles that can be expressed using four fundamental operations: the \emph{scramble shuffle} (\sh{ss}), the \emph{random cut} (\sh{rc}), the \emph{pile-scramble shuffle} (\sh{pss}), and the \emph{random pile cut} (\sh{rpc}). These operations are widely studied in the literature and form the basic toolbox for constructing card-based protocols.

\subsection{Related Work}
Card-based cryptography was first introduced by den Boer~\cite{5card} in 1990 through the well-known \textit{five-card trick}, which uses the \sh{rc} operation. In 1998, Niemi and Renvall~\cite{niemi} proposed a copy protocol that employs the \sh{ss}. In 2009, Mizuki and Sone~\cite{mizuki09} introduced the \textit{random bisection cut}, which can be viewed as a special case of both \sh{rpc} and \sh{pss} with two piles.

Subsequent works introduced additional shuffle techniques and theoretical results. In 2015, Koch et al.~\cite{nonclosed} introduced the notions of non-closed shuffles, whose set of permutations does not form a group, and non-uniform shuffles, which assign different probabilities to different permutations. Such shuffles are generally difficult to implement in practice. In the same year, Ishikawa et al.~\cite{scramble} introduced the \sh{pss}. Shinagawa et al.~\cite{polygon} introduced the \sh{rpc} in 2017 in the context of rotating polygon cards. Nishimura et al.~\cite{unequal} proposed special tools for performing shuffles with piles of unequal sizes, enabling certain non-closed shuffles to be realized. In a subsequent work, Nishimura et al.~\cite{nonuniform} further showed that the same tools can also realize certain non-uniform shuffles.

In 2020, Koch and Walzer~\cite{chosen} proved that every uniform closed shuffle can be realized using only \sh{rc} operations, although additional cards of the same size but with a different back color are required. More recently, Miyamoto and Shinagawa~\cite{graph} showed that certain shuffles corresponding to automorphism groups of graphs can be realized using only \sh{pss} operations, again requiring additional cards. Saito et al.~\cite{saito} proved that any shuffle with rational probabilities can be realized using only \sh{rc} operations with additional cards; however, the number of required cards can grow extremely large, especially when the outcome set contains more than two permutations.

On the experimental side, Miyahara et al.~\cite{time} experimentally measured the time required to perform some fundamental operations in practice and compared the running time of existing card-based protocols.

\subsection{Our Contribution}
In this paper, we introduce a hierarchy of shuffle operations used in card-based protocols based on their implementation complexity. We first analyze shuffle operations from both practical and theoretical viewpoints, identifying the fundamental operations that can be reliably performed in real-world settings.

Based on these observations, we classify shuffles into several levels according to the operations required to realize them. We prove separation results between these levels by showing that certain sets of permutations cannot be realized using only operations from lower levels. In particular, we establish that \sh{rc} is strictly more complex than \sh{ss}, and that \sh{rpc} is strictly more complex than \sh{rc}. We also show that some shuffles cannot be realized using the four fundamental operations commonly used in existing card-based protocols.

Finally, we perform an exhaustive search for small deck sizes to enumerate the permutation sets realizable at each level, revealing strong structural restrictions on achievable shuffles. Based on this hierarchy, we propose a new complexity measure for evaluating card-based protocols.

\section{Preliminaries}
In the formal computational model of card-based protocols proposed by Mizuki and Shizuya~\cite{formal}, a shuffle of a deck is mathematically defined by a pair $(\Pi, \mathscr{F})$, where $\Pi$ is a set of permutations and $\mathscr{F}$ is a probability distribution over $\Pi$. A~\emph{deterministic permutation} is a special case of a shuffle in which $|\Pi|=1$, i.e. $\Pi=\{\pi\}$ for some permutation $\pi$, and $\mathscr{F}$ assigns probability $1$ to $\pi$. Informally, this means the cards are rearranged according to a fixed rule without randomness.

We use the standard cycle notation for permutations. For example, $(1\,2)$ denotes the transposition that swaps positions $1$ and $2$, while $(1\,2\,3)$ denotes the cyclic permutation mapping $1\mapsto2$, $2\mapsto3$, and $3\mapsto1$. A permutation is called \emph{even} if it can be written as a composition of an even number of transpositions, and \emph{odd} otherwise. For example, $(1\,2)$ is odd, while $(1\,2\,3)$ is even.

A shuffle is called \emph{uniform} if $\mathscr{F}$ is the uniform distribution over $\Pi$, and is called \emph{closed} if $\Pi$ forms a subgroup of the symmetric group. Informally, a uniform shuffle selects every allowed permutation with equal probability, while a closed shuffle has the property that composing any two allowed permutations again produces an allowed permutation. In this paper, we mainly consider uniform shuffles. Unless stated otherwise, a shuffle denoted only by a set $\Pi$ refers to the shuffle $(\Pi, \mathscr{F})$, where $\mathscr{F}$ is the uniform distribution over $\Pi$.

Another property sometimes considered in the literature is the \emph{cyclic property}~\cite{abe}, which informally means that the set of allowed permutations consists of all powers of a single permutation. However, this property is irrelevant to the notion of \emph{shuffle complexity} studied in this paper. For example, the cyclic shuffle $\{\id, (1\,2)(3\,4)\}$ is more complex than the non-cyclic shuffle $\{\id, (1\,2), (3\,4), (1\,2)(3\,4)\}$ according to the hierarchy introduced later in this paper.

Next, we will formally define the four fundamental shuffle operations.

\subsection{Scramble Shuffle} \label{ssdef}
The scramble shuffle (\sh{ss}) rearranges a deck of $n$ cards into a uniformly random permutation, denoted by $S_n$, the symmetric group on $n$ elements. This operation is also sometimes called a \emph{perfect shuffle}. In practice, this corresponds to scrambling the cards randomly on the table. While this is a natural operation commonly performed, the first card-based protocol to formally employ this operation is the copy protocol of Niemi and Renvall~\cite{niemi}.

\subsection{Random Cut} \label{rcdef}
The random cut (\sh{rc}) rearranges a deck of $n$ cards into a uniformly random cyclic shift, denoted by $\{\id, \sigma, \sigma^2, \ldots, \sigma^{n-1}\}$, where $\sigma = (1\,2\,\ldots\,n)$. Here, $\sigma^i$ corresponds to shifting the deck cyclically by $i$ positions. In practice, this operation can be implemented by repeatedly performing a cut of the deck with randomly chosen cut positions~\cite{hindu}. This operation appears in the very first card-based protocol, the five-card trick of den Boer~\cite{5card}.

\subsection{Pile-Scramble Shuffle}
The pile-scramble shuffle (\sh{pss}) divides a deck of $kn$ cards into $k$ piles of $n$ consecutive cards, and then rearranges the $k$ piles according to a uniformly random permutation. Formally, we index the cards from $1$ to $kn$ according to their positions in the deck. The \sh{pss} is denoted by
\[
\{\pi_\tau \mid \tau \in S_k\},
\]
where $\pi_\tau((i-1)n+j) = (\tau(i)-1)n + j$ for $1 \le i \le k$ and $1 \le j \le n$.

In practice, this operation can be implemented by placing the $n$ cards of each pile into an envelope and scrambling the pile of $k$ envelopes. Alternatively, one may use a rubber band or paper clip to tie the cards in each pile together before scrambling them. This operation was first introduced by Ishikawa et al.~\cite{scramble}.

\subsection{Random Pile Cut}
The random pile cut (\sh{rpc}), also known as a \emph{pile-shifting shuffle}~\cite{makaro} or a \emph{random $k$-section cut}~\cite{equality}, divides a deck of $kn$ cards into $k$ piles of $n$ consecutive cards, and then rearranges the $k$ piles according to a uniformly random cyclic shift. Formally, it is denoted by
\[
\{\id, \sigma^n, \sigma^{2n}, \ldots, \sigma^{(k-1)n}\},
\]
where $\sigma = (1\,2\,\ldots\,kn)$.

Like the \sh{pss}, the \sh{rpc} can be implemented in practice by placing the $n$ cards of each pile into an envelope and cutting the pile of $k$ envelopes. Alternatively, one may use a rubber band or paper clip to tie the cards in each pile together before performing the cut. This operation was first introduced by Shinagawa et al.~\cite{polygon}.

In 2020, Koch and Walzer~\cite{chosen} showed that the \sh{rpc} can be implemented without envelopes, rubber bands, or paper clips; however, their method requires $k$ additional cards of the same size but with a different back color.

\section{Practical Viewpoint}
In this section, we discuss practical considerations that motivate our hierarchy of shuffle operations. The notion of complexity considered here is qualitative rather than quantitative: intuitively, a shuffle is considered more complex if it requires more sophisticated physical actions or additional tools to perform securely. A formal comparison criterion based on realizability will be introduced later in the theoretical viewpoint.

When designing card-based protocols, it is important to consider whether the shuffle operations used in the protocol can be performed reliably in practice. In particular, we aim to identify commonly used shuffle operations that are implementable using simple physical actions and verifiable by all players. Such operations should satisfy the following two properties.

\paragraph{1. The resulting permutation is unknown to everyone.}
After an operation is performed, no player should know the exact permutation applied to the deck, including the person performing the operation. This can be achieved through simple physical actions on the deck.

One common method is to scramble the cards on the table. In this process, the cards are mixed randomly by repeatedly moving them around on the table. In practice, the resulting distribution over permutations is not exactly uniform after finitely many scrambling actions. However, the process can be modeled as a Markov chain whose stationary distribution is uniform over all permutations. Thus, sufficiently long scrambling approximates the idealized \sh{ss} operation defined in Section~\ref{ssdef}.

Another natural action is to grab a portion of the deck and move it to another position. If the insertion position is chosen arbitrarily, repeated applications of this operation also converge to a distribution close to that produced by the \sh{ss}. However, if the operation is restricted so that cards are moved only in cyclic order (the bottom portion of the deck is moved to the top, or vice versa), the resulting permutation becomes a cyclic shift of the deck. In practice, a single cut does not necessarily produce a uniform distribution over cyclic shifts. However, if the deck is cut several times with independently chosen cut positions, the resulting process can be viewed as a Markov chain whose stationary distribution is uniform over all cyclic shifts. Thus, sufficiently many cuts approximate the idealized \sh{rc} operation defined in Section~\ref{rcdef}.

Thus, among simple physical actions that naturally hide the resulting permutation from all players, the two fundamental operations are the \sh{ss} and the \sh{rc}.

\paragraph{2. All players can verify that no cheating occurs.}
In addition to hiding the permutation, it must also be possible for all players to verify that the shuffler performs only permutations from the allowed set $\Pi$. If the shuffler were able to secretly apply permutations outside $\Pi$, the security of the protocol could be compromised.

From this perspective, the \sh{ss} is particularly easy to verify. Since the allowed set $\Pi$ consists of all permutations, every possible outcome of the scrambling process is valid. Therefore, players only need to observe that the cards are being mixed; there is no need to verify the exact permutation applied.

In contrast, the \sh{rc} is slightly more difficult to verify. In this case, the allowed set $\Pi$ consists only of cyclic shifts. Players must therefore ensure that the shuffler performs only a cyclic cut and does not secretly apply any other permutation. This requires closer observation during the operation.

\paragraph{Conclusion.}
Combining these two considerations, we conclude that the \sh{ss} is the easiest fundamental operation to perform and verify in practice, while the \sh{rc} is the second easiest. In contrast, the \sh{pss} and the \sh{rpc} typically require extra tools such as envelopes, rubber bands, paper clips, or cards of the same size but with a different back color, and therefore cannot be performed as easily using only the cards from a given deck.

\section{Theoretical Viewpoint}
We say that a shuffle $A$ is \emph{at least as complex as} a shuffle $B$, denoted by $A \succeq B$, if shuffle $B$ can be realized by applying shuffle $A$ several times, possibly interleaved with deterministic permutations, without using any additional cards or tools. We say that $A$ is \emph{strictly more complex than} $B$, denoted by $A \succ B$, if $A \succeq B$ but $B \not\succeq A$.

For example, the \sh{ss} can be realized using the \sh{rc}. Consider the following sequence of shuffles:
\[
\{\id, \sigma_2\},\ 
\{\id, \sigma_3, \sigma_3^2\},\ 
\ldots,\ 
\{\id, \sigma_n, \sigma_n^2, \ldots, \sigma_n^{n-1}\},
\]
where $\sigma_i = (1\,2\,\ldots\,i)$ for $2 \le i \le n$. Applying these shuffles sequentially produces a uniformly random permutation in $S_n$. Therefore, $\sh{rc} \succeq \sh{ss}$.

Similarly, the \sh{pss} can be realized using the \sh{rpc}. Since the \sh{rpc} performs a uniformly random cyclic shift on the $k$ piles, applying it $k-1$ times generates a uniformly random permutation of the $k$ piles. Hence, $\sh{rpc} \succeq \sh{pss}$.\footnote{The \sh{rpc} can also be realized using the \sh{pss} if a large number of additional cards are used~\cite{graph}.}

Moreover, since the \sh{ss} is a special case of the \sh{pss}, and the \sh{rc} is a special case of the \sh{rpc}, it follows that $\sh{pss} \succeq \sh{ss}$, and $\sh{rpc} \succeq \sh{rc}$.\footnote{The \sh{rpc} can also be realized using the \sh{rc} if additional cards of the same size but with a different back color are used~\cite{chosen}.}

The relationships between the four shuffle operations discussed above can be summarized as follows:
\[
\sh{rpc} \succeq \sh{pss} \succeq \sh{ss},
\qquad
\sh{rpc} \succeq \sh{rc} \succeq \sh{ss}.
\]

Next, we prove that $\sh{rc} \succ \sh{ss}$ by exhibiting a set of permutations realizable using the \sh{rc} but not realizable using only \sh{ss} operations and deterministic permutations.

\begin{theorem} \label{thm1}
For $n=3$, the shuffle $\{\id,(1\,2\,3),(1\,2\,3)^2\}$ cannot be realized using only \sh{ss} operations and deterministic permutations.
\end{theorem}

\begin{proof}
Suppose the only allowed randomized operations are the $2$-card \sh{ss} and the $3$-card \sh{ss}, together with deterministic permutations.

\begin{itemize}
    \item If a $3$-card \sh{ss} is used at least once, the distribution becomes uniform on $S_3$. Since the uniform distribution on $S_3$ is invariant under further deterministic permutations and \sh{ss} operations, the final distribution must remain uniform on all six permutations. Thus, the target set $\{\id,(1\,2\,3),(1\,2\,3)^2\}$ cannot be obtained.
    \item If no $3$-card \sh{ss} is used, then only $2$-card \sh{ss} operations occur. After $r$ such operations, each outcome has probability of the form $m/2^r$. Thus, a uniform distribution assigning probability $1/3$ to three outcomes cannot be obtained.
\end{itemize}

Hence, the shuffle $\{\id,(1\,2\,3),(1\,2\,3)^2\}$ cannot be realized.
\end{proof}

Therefore, $\sh{rc} \succ \sh{ss}$, which coincides with our conclusion from the practical viewpoint that the \sh{ss} is the easiest fundamental operation.

\subsection{Levels of Shuffles}
We now formally classify shuffles (with uniform distributions over sets of permutations) into levels according to the operations required to realize them.

\begin{itemize}
    \item Level~0 shuffles are those obtainable using only deterministic permutations. Therefore, they consist only of single-element sets.
    \item Level~1 shuffles are those obtainable using \sh{ss} operations and deterministic permutations.
    \item Level~2 shuffles are those obtainable using \sh{rc} operations and deterministic permutations.
\end{itemize}

For example, when $n=3$, the shuffle $\{\id,(1\,2\,3),(1\,2\,3)^2\}$ belongs to Level~2 since it can be realized using \sh{rc} operations but cannot be realized using only \sh{ss} operations.

Next, we prove that $\sh{rpc} \succ \sh{rc}$ by exhibiting a set of permutations realizable using the \sh{rpc} but not realizable using only \sh{rc} operations and deterministic permutations.

\begin{theorem}\label{thm2}
For $n=4$, the shuffle $\{\id,(1\,2)(3\,4)\}$ cannot be realized using only \sh{rc} operations and deterministic permutations.
\end{theorem}

\begin{proof}
Suppose the only allowed randomized operations are the $2$-card \sh{rc}, the $3$-card \sh{rc}, and the $4$-card \sh{rc}, together with deterministic permutations.

\begin{itemize}
    \item If a $2$-card or $4$-card \sh{rc} is used at least once, both even and odd permutations appear in the distribution (since a transposition is odd, and a $4$-cycle alternates parity among its powers). Since further deterministic permutations and \sh{rc} operations preserve the property that both parities appear, once both even and odd permutations appear they cannot be eliminated. However, the target set $\{\id,(1\,2)(3\,4)\}$ contains only even permutations, so it cannot be obtained.
    
    \item If no $2$-card or $4$-card \sh{rc} is used, then only $3$-card \sh{rc} operations occur. After $r$ such operations, each outcome has probability of the form $m/3^r$. Thus, a uniform distribution assigning probability $1/2$ to two outcomes cannot be obtained.
\end{itemize}

Hence, the shuffle $\{\id,(1\,2)(3\,4)\}$ cannot be realized.
\end{proof}

Indeed, the shuffle $\{\id,(1\,2)(3\,4)\}$ can be realized using the \sh{rpc} and deterministic permutations. For example, apply the deterministic permutation $(2\,3)$, followed by the \sh{rpc} with two piles of two cards, and then apply the deterministic permutation $(2\,3)$ again.

Therefore, $\sh{rpc} \succ \sh{rc}$, which also coincides with our conclusion from the practical viewpoint. We then define \emph{Level~3} shuffles to be those obtainable using \sh{rpc} operations and deterministic permutations. For example, when $n=4$, the shuffle $\{\id,(1\,2)(3\,4)\}$ belongs to Level~3.

\subsection{Shuffles Beyond the Four Fundamental Operations}
Some shuffles cannot be realized using the four fundamental operations. One example is the shuffle $\{\id,(1\,2\,3)\}$.

At first glance, one might attempt to implement this shuffle using an operation similar to the \sh{rpc} (or \sh{pss}), by dividing the deck into two piles of unequal size: one pile containing two cards and the other containing one card, and then randomly swapping the two piles. However, such an operation is not secure in practice. If rubber bands or paper clips are used to hold the piles together, players can distinguish the piles during the shuffle because they contain different numbers of cards. If envelopes are used instead, the players can still determine which pile is which after the cards are taken out of the envelopes. Therefore, this approach does not achieve a secure shuffle.

In fact, it is theoretically impossible to realize the shuffle $\{\id,(1\,2\,3)\}$ using only \sh{rpc} operations and deterministic permutations, as shown in Theorem~\ref{thm3}.

\begin{theorem} \label{thm3}
For $n=3$, the shuffle $\{\id,(1\,2\,3)\}$ cannot be realized using only \sh{rpc} operations and deterministic permutations.
\end{theorem}

\begin{proof}
Since a nontrivial \sh{rpc} operation requires at least two piles, each containing at least two cards, it can be applied to decks with at least four cards. Therefore, no nontrivial \sh{rpc} operation is available when $n=3$, and the model has essentially the same power as using only \sh{rc} operations and deterministic permutations.

Suppose the only allowed randomized operations are the $2$-card \sh{rc} and the $3$-card \sh{rc}, together with deterministic permutations.

\begin{itemize}
    \item If a $2$-card \sh{rc} is used at least once, both even and odd permutations appear in the distribution. Since further deterministic permutations and \sh{rc} operations preserve the property that both parities appear, once both even and odd permutations appear they cannot be eliminated. However, the target set $\{\id,(1\,2\,3)\}$ contains only even permutations, so it cannot be obtained.
    
    \item If no $2$-card \sh{rc} is used, then only $3$-card \sh{rc} operations occur. After $r$ such operations, each outcome has probability of the form $m/3^r$. Thus, a uniform distribution assigning probability $1/2$ to two outcomes cannot be obtained.
\end{itemize}

Hence, the shuffle $\{\id,(1\,2\,3)\}$ cannot be realized.
\end{proof}

Nishimura et al.~\cite{unequal} proposed a tool to help implement such shuffles, namely \sh{rpc} operations with piles of unequal sizes. We call this operation the \emph{unequal random pile cut} (\sh{urpc}). Their method uses boxes with sliding covers (see Fig.~\ref{boxes}). The cards of each pile are placed into a separate box, and the \sh{rc} is applied to the boxes. The boxes are then stacked on top of each other, and the covers between them are removed so that the cards fall into a single pile.

A simpler alternative, also proposed by~\cite{unequal}, is to use one envelope for each pile of cards together with one large envelope (see Fig.~\ref{envelopes}). After applying the \sh{rc} to the envelopes, all of them are placed inside the large envelope. The envelopes are then pulled out simultaneously, allowing the cards to fall into a single pile inside the large envelope.

\begin{figure}[H]
\centering
\begin{minipage}{0.45\textwidth}
    \centering
    \includegraphics[width=50mm]{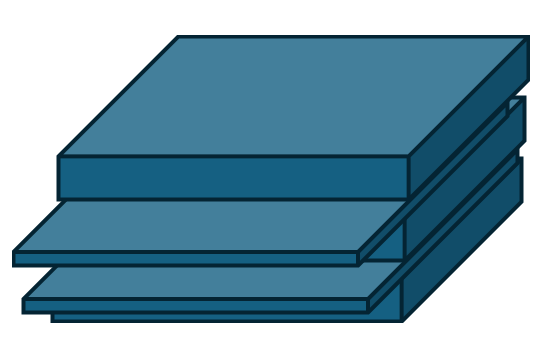}
    \caption{Boxes with sliding covers}
    \label{boxes}
\end{minipage}
\hfill
\begin{minipage}{0.45\textwidth}
    \centering
    \includegraphics[width=50mm]{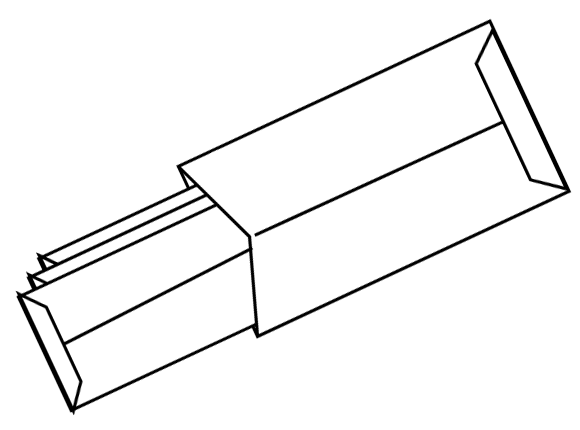}
    \caption{Small envelopes inside a large envelope}
    \label{envelopes}
\end{minipage}
\end{figure}

Theoretically, the \sh{urpc} allows us to perform shuffles of the form
\[
\{\id, \sigma^{a_1}, \sigma^{a_2}, \ldots, \sigma^{a_k}\},
\]
where $\sigma = (1\,2\,\ldots\,n)$, $k \le n-1$, and $a_1,a_2,\ldots,a_k$ are distinct elements of $\{1,2,\ldots,n-1\}$.

We define \emph{Level~4} shuffles to be those realizable using \sh{urpc} operations and deterministic permutations. For example, when $n=3$, the shuffle $\{\id,(1\,2\,3)\}$ belongs to Level~4.\footnote{More generally, any uniform shuffle can be realized using \sh{rc} operations with a sufficiently large number of additional cards, although the number of required cards may grow extremely large~\cite{saito}.}

The hierarchy of shuffle levels introduced above is summarized in Table~\ref{shufflelevels}. The hierarchy is cumulative: every Level~$i$ shuffle is also a Level~$i+1$ shuffle.

\begin{table}[H]
    \centering
    \caption{Hierarchy of shuffle levels}
    \label{shufflelevels}
    \begin{tabular}{|c|c|}
        \hline
        Level & Realizable using \\ \hline
        Level 0 & Deterministic permutations \\ \hline
        Level 1 & \sh{ss} operations and deterministic permutations \\ \hline
        Level 2 & \sh{rc} operations and deterministic permutations \\ \hline
        Level 3 & \sh{rpc} operations and deterministic permutations \\ \hline
        Level 4 & \sh{urpc} operations and deterministic permutations \\ \hline
    \end{tabular}
\end{table}

\subsection{Number of Realizable Shuffles}
For a deck of $n$ cards, there are $n!$ possible permutations, and therefore $2^{n!}-1$ possible nonempty subsets of permutations. By performing an exhaustive search over realizable subsets of permutations, we compute the number of realizable subsets at each level, as summarized in Table~\ref{table1}. As the table shows, only a small fraction of these subsets are realizable using the fundamental shuffle operations.

\begin{table}[H]
    \centering
	\caption{Number of realizable subsets using only each level of shuffles} \label{table1}
	\begin{tabular}{|c|c|c|c|c|c|c|}
		\hline
		$n$ & Level 0 & Level 1 & Level 2 & Level 3 & Level 4 & \#Nonempty subsets \\ \hline
		1 & 1 & 1 & 1 & 1 & 1 & 1 \\ \hline
		2 & 2 & 3 & 3 & 3 & 3 & 3 \\ \hline
		3 & 6 & 25 & 27 & 27 & 33 & 63 \\ \hline
		4 & 24 & 564 & 1,231 & 2,307 & 18,249 & 16,777,215 \\ \hline
	\end{tabular}
\end{table}

For $n=3$, Level~1 shuffles already generate $25$ subsets out of the $63$ possible nonempty subsets of $S_3$. The two additional subsets that become realizable at Level~2 but not at Level~1 are the two cyclic subgroups of size $3$: $\{\id,(1\,2\,3),(1\,3\,2)\}$ and $\{(1\,2),(1\,3),(2\,3)\}$.

Level~3 does not introduce any new subsets when $n\le3$. This is because a nontrivial \sh{rpc} operation requires at least two disjoint cycles, each of length at least two, which in turn requires at least four cards. Hence for $n\le3$, Level~3 is equivalent to Level~2.

For $n=4$, the number of realizable subsets grows significantly but still remains extremely small compared to the total number $2^{24}-1=16{,}777{,}215$ of possible nonempty subsets. This illustrates that the structural constraints imposed by the shuffle operations severely limit the sets of permutations that can be obtained.

\section{Complexity Measurement of Protocol}
Using the hierarchy of shuffle operations introduced in this paper, we propose a way to measure the shuffle complexity of card-based protocols. The goal is to evaluate protocols based on the types of shuffle operations they require.

Given a protocol, we represent its shuffle complexity by a tuple
\[
(a_1,a_2,a_3,a_4,a_5),
\]
where each component counts the number of fundamental operations of a particular type used in the protocol. Specifically:
\begin{itemize}
\item $a_1$ is the number of \sh{ss} operations,
\item $a_2$ is the number of \sh{rc} operations,
\item $a_3$ is the number of \sh{pss} or \sh{rpc} operations,
\item $a_4$ is the number of \sh{urpc} operations,
\item $a_5$ counts any other shuffle operations not covered by previous categories.
\end{itemize}

This representation is motivated by the observation that the practical difficulty of implementing shuffle operations can vary significantly depending on the type of shuffle and the method used to perform it. For example, the time required to perform an \sh{rpc} operation may differ greatly depending on whether envelopes, rubber bands, or additional cards are used. Therefore, representing complexity by counting the types of shuffle operations provides a more flexible way to compare protocols than assigning fixed time costs.

Previous experimental work by Miyahara et al.~\cite{time} attempted to measure the running time of card-based protocols by experimentally recording the exact time required to perform them in practice. Their approach focused on basic actions such as \texttt{add} (adding cards), \texttt{turn} (turning cards), \texttt{perm} (deterministic permutation), and \texttt{shuffle} (any shuffle). However, they did not explicitly distinguish between different types of shuffle operations, although they briefly noted that under some implementations, performing the random bisection cut can take more than four times as long as performing the \sh{rc}.

Our approach differs in that we explicitly distinguish different classes of shuffle operations according to the hierarchy introduced in this paper. Instead of measuring the exact execution time in seconds, we represent the shuffle complexity of a protocol as the tuple $(a_1,a_2,a_3,a_4,a_5)$, which records the number of operations of each type. This representation allows protocols to be compared in terms of the structural complexity of the shuffle operations they require.

\subsection{Shuffle Complexity of Selected Protocols}
Table~\ref{table2} summarizes the shuffle complexity of several well-known card-based protocols using the tuple representation introduced above. The table illustrates how different protocols rely on different types of shuffle operations. For example, the five-card trick~\cite{5card} and the equality function protocol of Shinagawa and Mizuki~\cite{6card} use only one \sh{rc}, while the Sudoku ZKP protocols of Sasaki et al.~\cite{sudoku2} rely heavily on \sh{ss} and \sh{pss} operations. Note that we only consider protocols with guaranteed finite running time, and the complexity shown corresponds to the worst-case analysis.

\begin{table}[H]
    \centering
	\caption{Shuffle complexity of selected card-based protocols} \label{table2}
	\begin{tabular}{|c|c|c|}
		\hline
		Protocol & Reference & Shuffle complexity \\ \hline
        Five-card trick & den Boer~\cite{5card}, 1990 & $(0,1,0,0,0)$ \\ \hline
        Six-card AND & Mizuki-Sone~\cite{mizuki09}, 2009 & $(0,0,1,0,0)$ \\ \hline
        Four-card XOR & Mizuki-Sone~\cite{mizuki09}, 2009 & $(0,0,1,0,0)$ \\ \hline
        Four-card AND & Mizuki et al.~\cite{mizuki12}, 2012 & $(1,0,1,0,0)$ \\ \hline
        Eight-card majority & Nishida et al.~\cite{majority}, 2013 & $(0,0,2,0,0)$ \\ \hline
        Five-card AND & Koch et al.~\cite[\S5]{nonclosed}, 2015 & $(3,0,1,2,0)$ \\ \hline
        Six-card equality & Shinagawa-Mizuki~\cite{6card}, 2018 & $(0,1,0,0,0)$ \\ \hline
        Six-card majority & Toyoda et al.~\cite{majority2}, 2021 & $(0,1,1,0,0)$ \\ \hline
        $2n$-card equality & Ruangwises-Itoh~\cite{equality}, 2021 & $(0,0,n-1,0,0)$ \\ \hline
        Sudoku ZKP (Protocol A) & Sasaki et al.~\cite[\S3.4]{sudoku2}, 2020 & $(n,0,4n,0,0)$ \\ \hline
        Sudoku ZKP (Protocol B) & Sasaki et al.~\cite[\S3.5]{sudoku2}, 2020 & $(2n,0,2n,0,0)$ \\ \hline
        Sudoku ZKP (Protocol C) & Sasaki et al.~\cite[\S4]{sudoku2}, 2020 & $(3n,0,1,0,0)$ \\ \hline
	\end{tabular}
\end{table}

The complexity tuples introduced in this section are intended primarily as descriptive summaries rather than as a total ordering of protocols. In general, two protocols with different tuples may not be directly comparable. Nevertheless, the hierarchy introduced in this paper suggests that higher-level shuffle operations should generally be regarded as more costly or difficult to implement than lower-level ones. Thus, in many practical situations, reducing the use of higher-level operations may be preferable even if it increases the number of lower-level operations.

\section{Conclusion and Future Work}
In this paper, we introduced a hierarchy of shuffle operations in card-based cryptography and studied their relationships from both practical and theoretical viewpoints. In particular, we showed that certain shuffles cannot be realized using lower-level operations, establishing strict separations within the hierarchy. We also proposed a tuple-based representation for describing the shuffle complexity of card-based protocols.

Possible future work includes finding explicit criteria for determining whether a given set of permutations is realizable within a particular level of the hierarchy. Another interesting direction is to derive formulas for the number of realizable subsets at each level for a given number of cards.

\subsubsection*{Acknowledgement}
The first author was supported by JSPS KAKENHI Grant Number JP25KJ1294. The second author was supported by the 111th Anniversary Engineering Research Catalyst Fund Towards U Top 100.


\begin{thebibliography}{}
	\bibitem{abe} Y. Abe, Y. Hayashi, T. Mizuki and H. Sone. Five-Card AND Computations in Committed Format Using Only Uniform Cyclic Shuffles. \textit{New Generation Computing}, 39(1): 97--114 (2021).
	\bibitem{makaro} X. Bultel, J. Dreier, J.-G. Dumas, P. Lafourcade, D. Miyahara, T. Mizuki, A. Nagao, T. Sasaki, K. Shinagawa and H. Sone. Physical Zero-Knowledge Proof for Makaro. In \textit{Proceedings of the 20th International Symposium on Stabilization, Safety, and Security of Distributed Systems (SSS)}, pp. 111--125 (2018).
	\bibitem{5card} B. den Boer. More Efficient Match-Making and Satisfiability: the Five Card Trick. In \textit{Proceedings of the Workshop on the Theory and Application of of Cryptographic Techniques (EUROCRYPT '89)}, pp. 208--217 (1990).
	\bibitem{sudoku} R. Gradwohl, M. Naor, B. Pinkas and G.N. Rothblum. Cryptographic and Physical Zero-Knowledge Proof Systems for Solutions of Sudoku Puzzles. \textit{Theory of Computing Systems}, 44(2): 245--268 (2009).
	\bibitem{scramble} R. Ishikawa, E. Chida and T. Mizuki. Efficient Card-Based Protocols for Generating a Hidden Random Permutation Without Fixed Points. In \textit{Proceedings of the 14th International Conference on Unconventional Computation and Natural Computation (UCNC)}, pp. 215--226 (2015).
	\bibitem{chosen} A. Koch and S. Walzer. Foundations for Actively Secure Card-Based Cryptography. In \textit{Proceedings of the 10th International Conference on Fun with Algorithms (FUN)}, pp. 17:1--17:23 (2020).
	\bibitem{nonclosed} A. Koch, S. Walzer and K. H\"{a}rtel. Card-Based Crypto-graphic Protocols Using a Minimal Number of Cards. In \textit{Proceedings of the 21st International Conference on the Theory and Application of Cryptology and Information Security (ASIACRYPT)}, pp. 783--807 (2015).
	\bibitem{time} D. Miyahara, I. Ueda, Y. Hayashi, T. Mizuki and H. Sone. Evaluating card-based protocols in terms of execution time. \textit{International Journal of Information Security}, 20(5): 729--740 (2021).
	\bibitem{graph} K. Miyamoto and K. Shinagawa. Graph Automorphism Shuffles from Pile-Scramble Shuffles. New Generation Computing, 40(1): 199--223 (2022).
	\bibitem{mizuki12} T. Mizuki, M. Kumamoto and H. Sone. The Five-Card Trick Can Be Done with Four Cards. In \textit{Proceedings of the 18th International Conference on the Theory and Application of Cryptology and Information Security (ASIACRYPT)}, pp. 598--606 (2012).
	\bibitem{formal} T. Mizuki and H. Shizuya. A formalization of card-based cryptographic protocols via abstract machine. \textit{International Journal of Information Security}, 13(1): 15--23 (2014).
	\bibitem{mizuki09} T. Mizuki and H. Sone. Six-Card Secure AND and Four-Card Secure XOR. In \textit{Proceedings of the 3rd International Frontiers of Algorithmics Workshop (FAW)}, pp. 358--369 (2009).
	\bibitem{niemi} V. Niemi and A. Renvall. Secure multiparty computations without computers. \textit{Theoretical Computer Science}, 191(1--2): 173--183 (1998).
	\bibitem{majority} T. Nishida, T. Mizuki and H. Sone. Securely Computing the Three-Input Majority Function with Eight Cards. In \textit{Proceedings of the 2nd International Conference on the Theory and Practice of Natural Computing (TPNC)}, pp. 193--204 (2013).
	\bibitem{nonuniform} A. Nishimura, Y. Hayashi, T. Mizuki and H. Sone. Pile-Shifting Scramble for Card-Based Protocols. \textit{IEICE Trans. Fundamentals}, E101.A(9): 1494--1502 (2018).
	\bibitem{unequal} A. Nishimura, T. Nishida, Y. Hayashi, T. Mizuki and H. Sone. Card-based protocols using unequal division shuffles. \textit{Soft Computing}, 22(2): 361--371 (2018).
	\bibitem{equality} S. Ruangwises and T. Itoh. Securely Computing the $n$-Variable Equality Function with $2n$ Cards. \textit{Theoretical Computer Science}, 887: 99--110 (2021).
	\bibitem{saito} T. Saito, D. Miyahara, Y. Abe, T. Mizuki and H. Shizuya. How to Implement a Non-uniform or Non-closed Shuffle. In \textit{Proceedings of the 9th International Conference on the Theory and Practice of Natural Computing (TPNC)}, pp. 107--118 (2020).
	\bibitem{sudoku2} T. Sasaki, D. Miyahara, T. Mizuki and H. Sone. Efficient card-based zero-knowledge proof for Sudoku. \textit{Theoretical Computer Science}, 839: 135--142 (2020).
	\bibitem{6card} K. Shinagawa and T. Mizuki. The Six-Card Trick: Secure Computation of Three-Input Equality. In \textit{Proceedings of the 21st Annual International Conference on Information Security and Cryptology (ICISC)}, pp. 123--131 (2018).
	\bibitem{polygon} K. Shinagawa, T. Mizuki, J.C.N. Schuldt, K. Nuida, N. Kanayama, T. Nishide, G. Hanaoka and E. Okamoto. Card-Based Protocols Using Regular Polygon Cards. \textit{IEICE Trans. Fundamentals}, E100.A(9): 1900--1909 (2017).
	\bibitem{majority2} K. Toyoda, D. Miyahara and T. Mizuki. Another Use of the Five-Card Trick: Card-Minimal Secure Three-Input Majority Function Evaluation. In \textit{Proceedings of the 22nd International Conference on Cryptology in India (INDOCRYPT)}, pp. 536--555 (2021).
	\bibitem{hindu} I. Ueda, D. Miyahara, A. Nishimura, Y. Hayashi, T. Mizuki and H. Sone. Secure implementations of a random bisection cut. \textit{International Journal of Information Security}, 19(4): 445--452 (2020).
\end{thebibliography}
\end{document}